%% file: l1l2-variations-dof.tex
\documentclass[smallextended,natbib,runningheads]{svjour3-mod}
\journalname{Annals of the Institute of Statistical Mathematics}

\usepackage[utf8]{inputenc}
\usepackage{mystyle}

\usepackage{graphicx}
\graphicspath{{./images/}}

\title{The degrees of freedom of the Group Lasso\\for a General Design
}

\author{Samuel~Vaiter \and
        Charles~Deledalle \and
        Gabriel~Peyr\'{e} \and
        Jalal~Fadili \and
        Charles~Dossal
}

\institute{Samuel~Vaiter, Gabriel~Peyr\'{e}  \at
              CEREMADE, CNRS, Universit\'{e} Paris-Dauphine, Place du Mar\'{e}chal De Lattre De Tassigny, 75775 Paris Cedex 16, France \\
              \email{\{samuel.vaiter,gabriel.peyre\}@ceremade.dauphine.fr}           
           \and
           Charles~Deledalle, Charles~Dossal \at
              IMB, CNRS, Universit\'{e} Bordeaux 1, 351, Cours de la lib\'{e}ration, 33405 Talence Cedex, France \\
              \email{\{charles.deledalle,charles.dossal\}@math.u-bordeaux1.fr}
           \and
           Jalal~Fadili \at
              GREYC, CNRS-ENSICAEN-Universit\'{e} de Caen, 6, Bd du Mar\'{e}chal Juin, 14050 Caen Cedex, France \\
              \email{Jalal.Fadili@greyc.ensicaen.fr}
}

\date{ }

\begin{document}

\maketitle

\input{sections/abstract.tex}

\input{sections/intro.tex}
\input{sections/local.tex}
\input{sections/gsure.tex}
\input{sections/proofs.tex}


\bibliographystyle{spbasic}
\bibliography{l1l2-variations-dof} 

\end{document}

%% file: sections/abstract.tex
\begin{abstract}
In this paper, we are concerned with regression problems where covariates can be grouped in nonoverlapping blocks, and where only a few of them are assumed to be active. In such a situation, the group Lasso is an attractive method for variable selection since it promotes sparsity of the groups. We study the sensitivity of any group Lasso solution to the observations and provide its precise local parameterization. When the noise is Gaussian, this allows us to derive an unbiased estimator of the degrees of freedom of the group Lasso. This result holds true for any fixed design, no matter whether it is under- or overdetermined. With these results at hand, various model selection criteria, such as the Stein Unbiased Risk Estimator (SURE), are readily available which can provide an objectively guided choice of the optimal group Lasso fit.

\keywords{Group Lasso \and Degrees of freedom \and Sparsity \and Model selection criteria}
\end{abstract}

%% file: sections/intro.tex
\section{Introduction}
\label{sec:introduction}

\subsection{Group Lasso}
\label{sub:intro-group}

Consider the linear regression problem
\begin{equation}\label{eq:linear-problem}
  y = \XX \xx_0 + \epsilon,
\end{equation}
where $y \in \RR^n$ is the response vector, $\xx_0 \in \RR^p$ is the unknown vector of regression coefficients to be estimated, $\XX \in \RR^{n \times p}$ is the design matrix whose columns are the $p$ covariate vectors, and $\epsilon$ is the error term. In this paper, we do not make any specific assumption on the number of observations $n$ with respect to the number of predictors $p$. Recall that when $n < p$, \eqref{eq:linear-problem} is an underdetermined linear regression model, whereas when $n \geq p$ and all the columns of $\XX$ are linearly independent, it is overdetermined. \\

Regularization is now a central theme in many fields including statistics, machine learning and inverse problems. It allows to reduce the space of candidate solutions by imposing some prior structure on the object to be estimated. This regularization ranges from squared Euclidean or Hilbertian norms \citep{Tikhonov97}, to non-Hilbertian norms that have sparked considerable interest in the recent years. Of particular interest are sparsity-inducing regularizations such as the $\lun$ norm which is an intensively active area of research, e.g. \citep{tibshirani1996regre,osborne2000new,donoho2006most,Candes09,BickelLassoDantzig07}; see \citep{BuhlmannVandeGeerBook11} for a comprehensive review. When the covariates are assumed to be clustered in a few active groups/blocks, the group Lasso has been advocated since it promotes sparsity of the groups, i.e. it drives all the coefficients in one group to zero together hence leading to group selection, see \citep{bakin1999adaptive,yuan2006model,bach2008consistency,Wei10} to cite a few.\\

Let $\Bb$ be a disjoint union of the set of indices i.e.~$\bigcup_{b \in \Bb} = \ens{1,\ldots,p}$ such that $b, b' \in \Bb, b \cap b' = \emptyset$.
For $\xx \in \RR^p$, for each $b \in \Bb$, $\xx_b=(\xx_i)_{i \in b}$ is a subvector of $\xx$ whose entries are indexed by the block $b$, and $|b|$ is the cardinality of $b$. The group Lasso amounts to solving
\begin{equation}\label{eq-group-lasso}\tag{\lassoP{y}{\lambda}}
  \xsoly(y) \in
  \uargmin{ \xx \in \RR^p }  
  \frac{1}{2}\norm{y - \XX \xx}^2 
  + \lambda \sum_{b \in \Bb} \norm{\xx_b} ,
\end{equation}
where $\lambda > 0$ is the regularization parameter and $\norm{\cdot}$ is the (Euclidean) $\ldeux$-norm. By coercivity of the penalty norm, the set of minimizers of \lasso is a nonempty convex compact set. Note that the Lasso is a particular instance of \eqref{eq-group-lasso} that is recovered when each block $b$ is of size 1.

\subsection{Degrees of Freedom}
\label{sub:intro-risk}

We focus in this paper on sensitivity analysis of any solution to \eqref{eq-group-lasso} with respect to the observations $y$ and the regularization parameter $\lambda$. This turns out to be a central ingredient to compute an estimator of the degrees of freedom (DOF) of the group Lasso response. The DOF is usually used to quantify the complexity of a statistical modeling procedure \citep{efron1986biased}.

More precisely, let $\widehat{\mu}(y)=\XX \xsoly(y)$ be the response or the prediction associated to an estimator $\xsoly(y)$ of $\xx_0$, and let $\mu_0 = \XX \xx_0$. We recall that $\widehat{\mu}(y)$ is always uniquely defined (see Lemma~\ref{lem:same-image}), although $\xsoly(y)$ may not as is the case when $\XX$ is a rank-deficient or underdetermined design matrix.
Suppose that $\epsilon$ is an additive white Gaussian noise $\epsilon \sim \Nn(0,\sigma^2 \Id_n)$. Following \citep{efron1986biased}, the DOF is given by
\begin{equation*}
\DOF = 
\sum_{i=1}^n \frac{\mathrm{cov}(y_i, \widehat{\mu}_i(y)])}{\sigma^2} ~.
\end{equation*}
The well-known Stein's lemma asserts that, if $\widehat{\mu}(y)$ is a weakly differentiable function for which
\[
\EE_\epsilon\pa{\left|\frac{\partial}{\partial y_i}\widehat{\mu}_i(y)\right|} < \infty ~,
\]
then its divergence is an unbiased estimator of its DOF, i.e.
\begin{equation*}
  \widehat{\DOF} = \mathrm{div}\widehat{\mu}(y) = \tr(\partial_y \widehat{\mu}(y)) \qandq \EE_\epsilon(\widehat{\DOF}) = \DOF ~,
\end{equation*}
where $\partial_y \widehat{\mu}(y)$ is the Jacobian of $\widehat{\mu}(y)$. It is well known that in Gaussian regression problems, an unbiased estimator of the DOF allows to get an unbiased of the prediction risk estimation $\EE_\epsilon \norm{\widehat{\mu}(y) - \mu_0}^2$ through e.g. the Mallow's $C_p$ \citep{mallows1973some}, the AIC \citep{akaike1973information} or the SURE \citep[Stein Unbiased Risk Estimate,][]{stein1981estimation}.
These quantities can serve as model selection criteria to assess the accuracy of a candidate model.

\subsection{Contributions} 
\label{sub:intro-contrib}

This paper establishes a general result (Theorem \ref{thm-local}) on local parameterization of any solution to the group Lasso \eqref{eq-group-lasso} as a function of the observation vector $y$. This local behavior result does not need $X$ to be full column rank. With such a result at hand, we derive an expression of the divergence of the group Lasso response. Using tools from semialgebraic geometry, we prove that this divergence formula is valid Lebesgue-almost everywhere (Theorem~\ref{thm-div}), and thus, this formula is a provably unbiased estimate of the DOF (Theorem~\ref{thm-dofsure}). In turn, this allows us to deduce an unbiased estimate of the prediction risk of the group Lasso through the SURE.

\subsection{Relation to prior works}

In the special case of standard Lasso with a linearly independent design, \citep{zou2007degrees} show that the number of nonzero coefficients is an unbiased estimate for the degrees of freedom. This work is generalized in \citep{2012-kachour-statsinica} to any arbitrary design matrix.
The DOF of the analysis sparse regularization (a.k.a.~generalized Lasso in statistics) is studied in \citep{tibshirani2012dof,vaiter-local-behavior}.

A formula of an estimate of the DOF for the group Lasso when the design is orthogonal within each group is conjectured in \citep{yuan2006model}. Its unbiasedness is proved but only for an orthogonal design. \citep{kato2009degrees} studies the DOF of a general shrinkage estimator where the regression coefficients are constrained to a closed convex set $C$. This work extended that of \citep{MeyerWoodroofe} which treats the case where $C$ is a convex polyhedral cone. When $\XX$ is full column rank, \citep{kato2009degrees} derived a divergence formula under a smoothness condition on the boundary of $C$, from which he obtained an unbiased estimator of the degrees of freedom. When specializing to the constrained version of the group Lasso, the author provided an unbiased estimate of the corresponding DOF under the same group-wise orthogonality assumption on $\XX$ as \citep{yuan2006model}. An estimate of the DOF for the group Lasso is also given by \citep{solo2010threshold} using  heuristic derivations that are valid only when $\XX$ is full column rank, though its unbiasedness is not proved. 

In \citep{vaiter-icml-workshops}, we derived an estimator of the DOF of the group Lasso and proved its unbiasedness when $\XX$ is full column rank, but without the orthogonality assumption required in \citep{yuan2006model,kato2009degrees}. In this paper, we remove the full column rank assumption, which enables us to tackle the much more challenging rank-deficient or underdetermined case where $p > n$.

\subsection{Notations} 
\label{sub:notations}

We start by some notations used in the rest of the paper.
We extend the notion of support, commonly used in sparsity by defining the \bsup $\bs(\xx)$ of $\xx \in \RR^n$ as
\begin{equation*}
  \bs(\xx) = \enscond{b \in \Bb}{\norm{\xx_b} \neq 0}.
\end{equation*}
The size of $\bs(\xx)$ is defined as $\abs{\bs(\xx)} = \sum_{b \in \Bb} \abs{b}$.
The set of all $\Bb$-supports is denoted $\Ii$.
We denote by $\XX_{I}$, where $I$ is a \bsup, the matrix formed by the columns $\XX_i$ where $i$ is an element of $b \in I$.
To lighten the notation in our derivations, we introduce the following block-diagonal operators
\begin{align*}
  &\delta_\xx : v \in \RR^{|I|} \mapsto ( v_b / \norm{\xx_b} )_{b \in I} \in \RR^{|I|} \\
  \qandq &P_\xx : v \in \RR^{|I|} \mapsto ( \mathrm{Proj}_{\xx_b^\bot}( v_b ) )_{b \in I} \in \RR^{|I|} ~,
\end{align*}
where $\mathrm{Proj}_{\xx_b^\bot} = \Id - \xx_b \xx_b^{\ins{T}}$ is the orthogonal projector on $\xx_b^\bot$.
For any matrix $A$, $A^{\ins{T}}$ denotes its transpose.

\subsection{Paper organization} 
\label{sub:organization}
The paper is organized as follows. Sensitivity analysis of the group Lasso solutions to perturbations of the observations is given in Section~\ref{sec:local}. Then we turn to the degrees of freedom and unbiased prediction risk estimation in Section~\ref{sec:dof}. The proofs are deferred to Section~\ref{sec:proofs} awaiting inspection by the interested reader.


%% file: sections/local.tex
\section{Local Behavior of the Group Lasso} 
\label{sec:local}

The first difficulty we need to overcome when $\XX$ is not full column rank is that $\xsoly(y)$ is not uniquely defined. Toward this goal, we are led to impose the following assumption on $\XX$ with respect to the block structure.

\paragraph{{\bf Assumption} $({{\rm\bf{A}}}(\beta))$}: Given a vector $\beta \in \RR^p$ of \bsup $I$, we assume that the finite subset of vectors $\enscond{\XX_b^{} \xx_b}{b \in I}$ is linearly independent.

\medskip

It is important to notice that $({\rm\bf{A}}(\beta))$ is weaker than imposing that $\XX_I$ is full column rank, which is standard when analyzing the Lasso. The two assumptions coincide for the Lasso, i.e. $|b|=1, \forall b \in I$.\\


Let us now turn to sensitivity of the minimizers $\xsoly(y)$ of \lasso to perturbations of $y$. Toward this end, we will exploit the fact that $\xsoly(y)$ obeys an implicit parameterization. But as optimal solutions turns out to be not everywhere differentiable, we will concentrate on a local analysis where  $y$ is allowed to vary in a neighborhood where non-differentiability will not occur. This is why we need to introduce the following transition space $\Hh$.

\begin{defn}\label{defn:h}
  Let $\lambda > 0$.
  The \emph{transition space} $\Hh$ is defined as
  	\begin{align*}
    		\Hh = \bigcup_{I \subset \Ii} \bigcup_{b \not\in I} \; \Hh_{I,b},
	\qwhereq \Hh_{I,b} &= \bd(\pi(\Aa_{I,b})) ,
	\end{align*}
	where we have denoted  
	\eq{
		\pi : \RR^n \times \RR^{I,*} \times \RR^{I,*} \to \RR^n
		\qwhereq \RR^{I,*} = \prod_{b \in I} (\RR^{|b|} \setminus \{0\})
	}
	the canonical projection on $\RR^n$ (with respect to the first component), $\bd C$ is the boundary of the set $C$, and 
\begin{align*}
  \Aa_{I,b} =
  \Big\{ &(y, \xx_I, v_I) \in \RR^n \times \RR^{I,*} \times \RR^{I,*} \; \setminus \\
    &\norm{\XX_b^{\ins{T}} (y - \XX_I \xx_I) } = \lambda, \\
    &\XX_I^{\ins{T}} (\XX_I \xx_I - y) + \lambda v_I = 0, \\
    &\forall g \in I, v_g = \frac{\xx_g}{\norm{\xx_g}}
  \Big\} ~.
\end{align*}
\end{defn}

We are now equipped to state our main sensitivity analysis result.
\begin{thm}\label{thm-local}
  Let $\lambda > 0$.
  Let $y \not\in \Hh$, and $\xsoly(y)$ a solution of \lasso.
  Let $I = \bs(\xsoly(y))$ be the \bsup of $\xsoly(y)$ such that $({\rm\bf{A}}(\xsoly(y)))$ holds.
  Then, there exists an open neighborhood of $y$ $\Oo \subset \RR^n$, and a mapping $\solm : \Oo \to \RR^p$ such that
  \begin{enumerate}
    \item For all $\bar y \in \Oo$, $\solmB$ is a solution of $(\lassoB)$, and $\solm(y) = \xsoly(y)$.
    \item the \bsup of $\solmB$ is constant on $\Oo$, i.e.~
    \begin{equation*}
      \forall \bar y \in \Oo, \quad \bs(\solmB) = I,
    \end{equation*}
    \item the mapping $\solm$ is $\Cc^1(\Oo)$ and its Jacobian is such that $\forall \bar y \in \Oo$,
      \begin{align}\label{eq-differential}
        \partial_{\bar y} \solm_{I^c}(\bar y) & = 0 \qandq
        \partial_{\bar y} \solm_I(\bar y) = d(y,\lambda) \\
      \qwhereq d(y,\lambda) & = \bpa{ \XXX + \lambda \delta_{\xsoly(y)}  \circ P_{\xsoly(y)}  }^{-1} \XX_I^{\ins{T}} \\
      \qandq I^c & = \enscond{b \in \Bb}{b \notin I}.
    \end{align}
  \end{enumerate}
\end{thm}


%% file: sections/gsure.tex
\section{Degrees of freedom and Risk Estimation}
\label{sec:dof}

As remarked earlier and stated formally in Lemma~\ref{lem:same-image}, all solutions of the Lasso share the same image under $\XX$, hence allowing us to denote the prediction $\widehat{\mu}(y)$ without ambiguity as a single-valued mapping. The next theorem provides a closed-form expression of the local variations of $\widehat{\mu}(y)$ with respect to the observation $y$. In turn, this will yield an unbiased estimator of the degrees of freedom and of the prediction risk of the group Lasso.
\begin{thm}\label{thm-div}
  Let $\lambda > 0$.
  For all $y \not\in \Hh$, there exists a solution $\xsoly(y)$ of \lasso with \bsup $I = \bs(\xsoly(y))$ such that $({\rm\bf{A}}(\xsoly(y)))$ is fulfilled. 
  Moreover, The mapping $y \mapsto \widehat{\mu}(y)=X\xsoly(y)$ is $\Cc^1(\RR^n \setminus \Hh)$ and,
  \begin{equation}\label{eq:diverg}
    \diverg(\widehat{\mu}(y)) = \tr (\XX_I d(y,\lambda))
  \end{equation}
  where $\xsoly(y)$ is such that $({\rm\bf{A}}(\xsoly(y)))$ holds.
\end{thm}

\begin{thm}\label{thm-dofsure}
Let $\lambda > 0$. Assume $y = \XX \xx_0 + \epsilon$ where $\epsilon \sim \Nn(0,\sigma^2 \Id_n)$. The set $\Hh$ has Lebesgue measure zero, and therefore \eqref{eq:diverg} is an unbiased estimate of the DOF of the group Lasso. Moreover, an unbiased estimator of the prediction risk $\EE_\epsilon \norm{\widehat{\mu}(y) - \mu_0}^2$ is given by the $\SURE$ formula
\begin{align}
\label{eq:sure}
    \SURE(\widehat{\mu}(y)) = &
    \norm{y - \widehat{\mu}(y)}^2
    - n \sigma^2
    + 2 \sigma^2 \tr (\XX_I d(y,\lambda))~.
\end{align}
\end{thm}


Although not given here explicitly, Theorem~\ref{thm-dofsure} can be straightforwardly extended to
unbiasedly of measures of the risk, including the {\it projection} risk, or the {\it estimation}
risk (in the full rank case) through the Generalized Stein Unbiased Risk Estimator as proposed in \citep{vaiter-local-behavior}.\\

An immediate corollary of Theorem~\ref{thm-dofsure} is obtained when $\XX$ is orthogonal, and without loss of generality $\XX=\Id_n$, i.e. $\widehat{\mu}(y)$ is the block soft thresholding estimator. We then recover the expression found by~\citep{yuan2006model}.
\begin{cor}\label{cor-particular}
  If $\XX = \Id_n$, then
  \begin{equation*}
    \widehat{\DOF} = \abs{I} - \lambda \sum_{b \in I} \dfrac{\abs{b} - 1}{\norm{y_b}}
  \end{equation*}
  where $I = \bigcup \enscond{b \in \Bb}{\norm{y_b} > \lambda}$.
  Moreover, the $\SURE$ is given by
  \begin{align*}
    \SURE(\widehat{\mu}(y)) = &
    -n \sigma^2
    +
    (2 \sigma^2 + \lambda^2) \abs{I}
    + \sum_{b \notin I} \norm{y_b}^2
    - 2 \sigma^2 \lambda \sum_{b \in I} \dfrac{\abs{b} - 1}{\norm{y_b}} ~.
  \end{align*}
\end{cor}

\newif\ifREL
\RELtrue
\ifREL
We finally quantify the (relative) reliability of the SURE by computing the expected squared-error between $\SURE(\widehat{\mu}(y))$ and the true squared-error
\eq{
\SE(\widehat{\mu}(y)) = \norm{\widehat{\mu}(y)-\mu_0}^2 ~.
}
\begin{prop}\label{prop:rel}
Under the assumptions of Theorem~\ref{thm-dofsure}, the relative reliability obeys
\eq{
\EE_w\left[\frac{\pa{\SURE(\widehat{\mu}(y))-\SE(\widehat{\mu}(y))}^2}{n^2\sigma^4}\right] \leq \frac{18+4\EE_w\pa{\norm{U_I}^2}}{n} + \frac{8\norm{\mu_0}^2}{n^2\sigma^2} ~.
}
\eq{
	\qwhereq U_I = X_I^\ins{T}X_I \bpa{ \XXX + \lambda \delta_{\xsoly(y)}  \circ P_{\xsoly(y)}  }^{-1}.
}
In particular, it decays at the rate  $O(1/n)$ if $\EE_w\pa{\norm{U_I}^2}=O(1)$.
\end{prop}

Note that when $X=\Id_n$, the proof of Corollary~\ref{cor-particular} yields that $\norm{ U_I }=1$.

\fi


%% file: sections/proofs.tex
\section{Proofs} 
\label{sec:proofs}

This section details the proofs of our results.
For a vector $\xx$ whose \bsup is $I$, we introduce the following normalization operator
\begin{equation*}
  \no(\xx_I) = v_I \qwhereq \forall b \in I, v_b = \frac{\xx_b}{\norm{\xx_b}} .
\end{equation*} 

\subsection{Preparatory lemmata}
By standard arguments of convex analysis and using the subdifferential of the group Lasso $\lun-\ldeux$ penalty, the following lemma gives the first-order sufficient and necessary optimality condition of a minimizer of \lasso; see e.g. \cite{bach2008consistency}.
\begin{lem}\label{lem:first-order}
  A vector $\xxs \in \RR^p$ is a solution of \lasso ~if, and only if the following holds
  \begin{enumerate}
    \item On the \bsup $I = \bs(\xxs)$,
    \begin{equation*}
      \XX_I^{\ins{T}}(y - \XX_I \xxs_I) = \lambda \no(\xxs_I) .
    \end{equation*}
    \item For all $b \in \Bb$ such that $b \not\in I$, one has
    \begin{equation*}
      \norm{\XX_b^{\ins{T}}(y - \XX_I \xxs_I)} \leq \lambda .
    \end{equation*}
  \end{enumerate}
\end{lem}

\medskip

We now show that all solutions of \lasso share the same image under the action of $\XX$, which in turn implies that the prediction/response vector $\widehat{\mu}$ is a single-valued mapping of $y$.
\begin{lem}\label{lem:same-image}
  If $\xx^0$ and $\xx^1$ are two solutions of \lasso, then $\XX \xx^0 = \XX \xx^1$.
\end{lem}
\begin{proof}
  Let $\xx^0, \xx^1$ be two solutions of \lasso ~such that $\XX \xx^0 \neq \XX \xx^1$.
  Take any convex combination $\xx^\rho = (1-\rho) \xx^0 + \rho \xx^1$, $\rho \in ]0,1[$.
  Strict convexity of $u \mapsto \norm{y - u}^2$ implies that the Jensen inequality is strict, i.e.
  \begin{equation*}
    \dfrac{1}{2} \norm{y - \XX \xx^\rho}^2 
    <  
    \dfrac{1-\rho}{2} \norm{y - \XX \xx^0}^2 + \dfrac{\rho}{2} \norm{y - \XX \xx^1}^2 ~.
  \end{equation*}
  Denote the $\lun-\ldeux$ norm $\normg{\beta}=\sum_{b \in Bb} \norm{\beta_b}$. Jensen's inequality applied to $\normg{\cdot}$ gives
  \begin{equation*}
    \normg{\xx^\rho} \leq (1-\rho) \normg{\xx^0} + \rho \normg{\xx^1} ~.
  \end{equation*}
  Summing these two inequalities we arrive at $\dfrac{1}{2} \norm{y - \XX \xx^\rho}^2 + \la\normg{\xx^\rho} < \dfrac{1}{2} \norm{y - \XX \xx^0}^2 + \la\normg{\xx^0}$, a contradiction since $\xx^0$ is a minimizer of \lasso.
\end{proof}

\input{sections/proofs-local}
\input{sections/proofs-dof}

%% file: sections/proofs-local.tex
\subsection{Proof of Theorem~\ref{thm-local}} 
\label{sub:local}

We first need the following lemma.
\begin{lem}\label{lem-invertible}
  Let $\xx \in \RR^p$ and $\lambda > 0$.
  Assume that $({\rm\bf{A}}(\xx))$ holds for $I$ the \bsup of $\xx$. 
  Then $\XXX+ \lambda \delta_\xx \circ P_\xx$ is invertible.
\end{lem}
\begin{proof}
We prove that $\XXX+ \lambda \delta_\xx \circ P_\xx$ is actually symmetric definite positive. First observe that $\XXX$ and $\delta_\xx \circ P_\xx$ are both symmetric semidefinite positive. Indeed, $\delta_\xx$ is diagonal (with strictly positive diagonal entries), and $P_\xx$ is symmetric since it is a block-wise orthogonal projector, and we have
\eq{
\dotp{x}{\delta_\xx \circ P_\xx(x)} = \sum_{b \in I} \frac{\norm{\mathrm{Proj}_{\xx_b^\bot}(x)}^2}{\norm{\xx_b}} \geq 0, \quad \forall x \in \RR^{|I|} ~.
}
The inequality becomes an equality if and only if $x=\xx_I$, i.e. $\Ker \delta_\xx \circ P_\xx = \{\xx_I\}$. 

It remains to show that $\Ker \XXX \cap \Ker \delta_\xx \circ P_\xx = \ens{0}$. Suppose that $\xx_I \in \Ker \XXX$. This is equivalent to $\xx_I \in \Ker \XX_I$ since
\eq{
\dotp{\xx_I}{\XXX \xx_I} = \norm{\XX_I\xx_I}^2 ~.
}
But this would mean that
\eq{
\XX_I \beta_I = \sum_{b \in I} \XX_b \xx_b = 0
}
which is in contradiction with the linear independence assumption $({\rm\bf{A}}(\xx))$. \qed

%
\end{proof}

\medskip

Let $y \not\in \Hh$.
We define $I = \bs(\xsoly(y))$ the \bsup of a solution $\xsoly(y)$ of \lasso.
We define the following mapping
\begin{equation*}
  \Gamma(\xx_I,y) = \XX_I^{\ins{T}}(\XX_I^{} \xx_I^{} - y) + \lambda \no(\xx_I) .
\end{equation*}
Observe that the first statement of Lemma~\ref{lem:first-order} is equivalent to $\Gamma(\xsoly_I(y),y) = 0$.

Any $\xx_I \in \RR^{\abs{I}}$ such that $\Gamma(\xx_I,y) = 0$ is solution of the problem
\begin{equation}\label{eq:restricted}\tag{$\lassoP{y}{\lambda}_I$}
  \umin{\xx_I \in \RR^{\abs{I}}}
  \frac{1}{2} \norm{y - \XX_I \xx_I}^2 
  + \lambda  \sum_{g \in I} \norm{\xx_g} ~.
\end{equation}

Our proof will be split in three steps.
We first prove the first statement by showing that there exists a mapping $\bar y \mapsto \solmB$ and an open neighborhood $\Oo$ of $y$ such that every element $\bar y$ of $\Oo$ satisfies $\Gamma(\solm_I(\bar y), \bar y) = 0$ and $\solm_{I^c}(\bar y) = 0$.
Then, we prove the second assertion that $\solmB$ is a solution of ($\lassoP{\bar y}{\lambda}$) for $\bar y \in \Oo$.
Finally, we obtain ~\eqref{eq-differential} from the implicit function theorem.

\begin{enumerate}
\item The Jacobian of $\Gamma$ with respect to the first variable reads on $\RR^{I,*} \times \RR^n$
\begin{equation*}
  \partial_1 \Gamma(\xx_I,y) = 
  \XXX + \lambda \delta_{\xx_I} \circ P_{\xx_I}.
\end{equation*}
The mapping $\partial_1 \Gamma$ is invertible according to Lemma~\ref{lem-invertible}.
Hence, using the implicit function theorem, there exists a neighborhood $\widetilde \Oo$ of $y$ such that we can define a mapping $\solm_I : \widetilde \Oo \to \RR^{\abs{I}} $ which is $\Cc^1(\widetilde \Oo)$, and satisfies for $\bar y \in \widetilde \Oo$
\begin{equation*}
    \Gamma(\solm_I(\bar y),\bar y)  =  0
    \qandq
    \solm_I(y) = \xsoly_I(y).
\end{equation*}
We then extend $\solm_I$ on $I^c$ as $\solm_{I^c}(\bar y)=0$, which defines a continuous mapping $\solm : \widetilde \Oo \to \RR^p$.

\item From the second minimality condition of Lemma~\ref{lem:first-order}, we have
\begin{equation*}
  \forall b \notin I, \quad \norm{\XX_b^{\ins{T}}(y - \XX_I \xsoly_I(y))} \leq \lambda .
\end{equation*}
We define the two following sets
\begin{gather*}
  J_{\text{sat}} = \enscond{b \not\in I}{\norm{\XX_b^{\ins{T}} (y - \XX_I \xsoly_I(y))} = \lambda} , \\
  J_{\text{nosat}} = \enscond{b \not\in I}{\norm{\XX_b^{\ins{T}} (y - \XX_I \xsoly_I(y))} < \lambda} ,
\end{gather*}
which forms a disjoint union of $I^c = J_{\text{sat}} \cup J_{\text{nosat}}$.

\begin{enumerate}[label=\alph{*}), ref=\alph{*})]
\item By continuity of $\bar y \mapsto \solm_I(\bar y)$ and since $\solm_I(y) = \xsoly_I(y)$, we can find a neighborhood $\Oo$ of $y$ included in $\widetilde \Oo$ such that
\begin{equation*}
    \foralls \bar y \in \Oo,\,
    \foralls b \in J_{\text{nosat}}, \quad
    \norm{ \XX_b^{\ins{T}}( \bar y - \XX_I \solm_I(\bar y) ) } \leq \lambda .
\end{equation*}

\item Consider now a block $b \in J_{\text{sat}}$.
Observe that the vector $(y, \xsoly_I(y), \no(\xsoly_I(y)))$ is an element of $\Aa_{I,b}$.
In particular $y \in \pi(\Aa_{I,b})$.
Since by assumption $y \not\in \Hh$, one has $y \not\in \bd(\pi(\Aa_{I,b}))$.
Hence, there exists an open ball $\mathbb{B}(y, \epsilon)$ for some $\epsilon > 0$ such that $\mathbb{B}(y, \epsilon) \subset \pi(\Aa_{I,b})$.
Notice that every element of $\bar y \in \mathbb{B}(y, \epsilon)$ is such that there exists $(\bar\xx_I, \bar v_I) \in \RR^{I,*} \times \RR^{I,*}$ with
\begin{align*}
   \norm{\XX_b^{\ins{T}} (\bar y - \XX_I \bar \xx_I) } &= \lambda \\
   \XX_I^{\ins{T}} (\XX_I \bar \xx_I - \bar y) + \lambda \bar v_I &= 0 \\
   \bar v_I &= \no(\bar\xx_I) ~.
\end{align*}
Using a similar argument as in the proof of Lemma~\ref{lem:same-image}, it is easy to see that all solutions of~\eqref{eq:restricted} share the same image under $\XX_I$. Thus the vector $(\bar y, \solm_I(\bar y), \no(\solm_I(\bar y)))$ is an element $\Aa_{I,b}$, and we conclude that 
\begin{equation*}
  \forall \bar y \in \mathbb{B}(y, \epsilon), \quad
  f_{I,b}(\bar y) = \norm{\XX_b^{\ins{T}} (\bar y - \XX_I\solm_I(\bar y)) } = \lambda .
\end{equation*}
Hence, $f_{I,b}$ is locally constant around $y$ on an open ball $\bar \Oo$.

Moreover, by definition of the mapping $\solm_I$, one has for all $\bar y \in \Oo \cap \bar \Oo$
\begin{equation*}
  \XX_I^{\ins{T}}(y - \XX_I \solm_I(\bar y)) = \lambda \no(\solm_I(\bar y)) \text{ and } \bs(\solm_I(\bar y)) = I .
\end{equation*}
According to Lemma~\ref{lem:first-order}, the vector $\solm(\bar y)$ is a solution of ($\lassoP{\bar y}{\lambda}$).

\end{enumerate}

\item By virtue of statement 1., we are in position to use the implicit function theorem, and we get the Jacobian of $\solm_I$ as
\begin{equation*}
    \partial_{\bar y} \solm_I(\bar y) = -
    \big( \partial_1 \Gamma(\solm_I(y),y) \big)^{-1} 
    \big( \partial_2 \Gamma(\solm_I(y),y) \big)        
\end{equation*}
where $\partial_2 \Gamma(\solm_I(y),y)= \XX_I^{\ins{T}}$, which leads us to~\eqref{eq-differential}.

\end{enumerate}

%% file: sections/proofs-dof.tex
\subsection{Proof of Theorem \ref{thm-div}} 
\label{sub:dof}

We define the set
\begin{equation}\label{eq:cnd}
   \cnd_{I} =
   \enscond{\xx_I\in \RR^{\abs{I}}}
   {\forall \mu \in \RR^{\sharp {I}}, \sum_{i = 1}^{\sharp I} \mu_i \XX_{b_i}^{} \xx_{b_i} = 0 \Rightarrow \mu = 0} ~,
\end{equation}
where $\sharp {I}$ is the number of blocks in $I$, and $b_i \in I$ is the $i$-th block in $I$.
It is easy to see that $\xxs \in \cnd_I$ for $I$ the \bsup of $\xxs$ if and only if $({\rm\bf{A}}(\xxs))$.

\medskip

The following lemma proves that there exists a solution $\xxs$ of \lasso such that $({\rm\bf{A}}(\xxs))$ holds. A similar result with a different proof can be found in \citep{LiuZhang09}.

\begin{lem}
  There exists a solution $\xxs$ of \lasso such that $\xxs \in \cnd_I$ where $I = \bs(\xxs)$.
\end{lem}
\begin{proof}
  Let $\xx^0$ be a solution of \lasso and $I = \bs(\xx^0)$ such that $\xx_I^0 \not\in \cnd_I$.
  There exists $\mu \in \RR^{\sharp {I}}$ such that
  \begin{equation}\label{eq:inde}
    \sum_{i = 1}^{\sharp I} \mu_i \XX_{b_i} \xx_{b_i}^0 = 0 .
  \end{equation}
  Consider now the family $t \mapsto \xx^t$ defined for every $t \in \RR$
  \begin{equation}
    \forall b_i \in I, \quad \xx_{b_i}^t = (1 + t \mu_i) \xx_{b_i}^0 \qandq \beta^t_{I^c} = 0 ~.
  \end{equation}
  Consider $t_0 = \min \enscond{\abs{t} \in \RR}{\exists b_i \in I \text{ such that } 1 + t \mu_i = 0} $. Without loss of generality, we assume that $t_0>0$.
  Remark that for all $t \in [0,t_0)$, $\xx^t$ is a solution of \lasso.
  Indeed, $I$ is the \bsup of $\xx^t$ and
  \begin{equation}
    \XX_I \xx_I^t = \XX_I \xx_I^{0} + t \underbrace{\sum_{i=1}^{\sharp I} \mu_i \XX_{b_i} \xx_{b_i}^0}_{=0 \text{ using~\eqref{eq:inde}} } = \XX_I \xx_I^{0} .
  \end{equation}
  Hence, 
  \begin{equation*}
    \XX_I^{\ins{T}}(y - \XX_I \xx_I^t) 
    = \XX_I^{\ins{T}}(y - \XX_I \xx_I^{0})
    = \lambda \no(\xx_I)
    = \lambda \no(\xx_I^t) ,
  \end{equation*}
  and
  \begin{equation*}
    \norm{\XX_b^{\ins{T}}(y - \XX_I \xx_I^t)} = \norm{\XX_b^{\ins{T}}(y - \XX_I \xx^{0}_I)} \leq \lambda , \quad \forall b \in I^c ~.
  \end{equation*}
  Since the image of all solutions of \lasso are equal under $\XX$, one has
  \begin{equation*}
    \XX \xx^{t} = \XX \xx^{0}
    \qandq
    \normg{\xx^{t}} = \normg{\xx^0} .    
  \end{equation*}
  where $\normg{\cdot}$ is the $\lun-\ldeux$ norm. Consider now the vector $\xx^{t_0}$.
  By continuity of $\xx \mapsto \XX \xx$ and $\xx \mapsto \normg{\xx}$, one has
  \begin{equation*}
    \XX \xx^{t_0} = \XX \xx^{0}
    \qandq
    \normg{\xx^{t_0}} = \normg{\xx^0} .
  \end{equation*}
  Hence, $\xx^{t_0}$ has a \bsup $I_{t_0}$ strictly included in $I$ (in the sense that for all $b \in I_{t_0}$ one has $b \in I$) and is a solution of \lasso.
  Iterating this argument with $\xx^0 = \xx^{t_0}$ shows that there exists a solution $\xxs$ such that $\xxs \in \cnd_{\bs(\xxs)}$. This concludes the proof of the lemma. \qed
\end{proof}
  
According to Theorem \ref{thm-local}, $y \mapsto \xsoly(y)$ is $\Cc^1(\RR^n \setminus \Hh)$. This property is preserved under the linear mapping $\XX$ which shows that $\widehat{\mu}$ is also $\Cc^1(\RR^n \setminus \Hh)$. Thus, taking the trace of the Jacobian $X_I d(y,\lambda)$ gives the divergence formula~\eqref{eq:diverg} for any solution $\xsoly(y)$ such that $({\rm\bf{A}}(\xsoly(y)))$ holds.

\subsection{Proof of Theorem \ref{thm-dofsure}} 
\label{sub:sure}

The next lemma shows that the transition space has zero measure.

\begin{lem}
  Let $\lambda > 0$.
  The transition space $\Hh$ is of zero measure with respect to the Lebesgue measure of $\RR^n$.
\end{lem}

\begin{proof}
We obtain this result by proving that all $\Hh_{I,b}$ are of zero measure for all $I$ and $b \not\in I$, and that the union is over a finite set.

We recall from \citep{coste2002intro} that any semialgebraic set $S \subseteq \RR^n$ can be decomposed in a disjoint union of $q$ semialgebraic subsets $C_i$ each diffeomorphic to $(0,1)^{d_i}$.
The dimension of $S$ is thus
\begin{equation*}
  d = \umax{i \in \{1,\dots,q\}} d_i \leq n.
\end{equation*}
The set $\Aa_{I,b}$ is an algebraic, hence a semialgebraic, set.
By the fundamental Tarski-Seidenberg principle, the canonical projection $\pi(\Aa_{I,b})$ is also semialgebraic.
The boundary $\bd(\pi(\Aa_{I,b}))$ is also semialgebraic with a strictly smaller dimension than $\pi(\Aa_{I,b})$
\begin{equation*}
  \dim \Hh_{I,b}
  =
  \dim \bd(\pi(\Aa_{I,b}))
  <
  \dim \pi(\Aa_{I,b})
  \leq
  n
\end{equation*}
whence we deduce that $\Hh$ is of zero measure with respect to the Lebesgue measure on $\RR^n$. \qed
\end{proof}

As $\widehat{\mu}$ is uniformly Lipschitz over $\RR^n$, using similar arguments as in \citep{MeyerWoodroofe}, we get that $\widehat{\mu}$ is weakly differentiable with an essentially bounded gradient. Moreover, the divergence formula~\eqref{eq:diverg} holds valid almost everywhere, except on the set $\Hh$ which is of Lebesgue measure zero. We conclude by invoking Stein's lemma \citep{stein1981estimation} to establish unbiasedness of the estimator $\widehat \DOF$ of the DOF.

%

Plugging the DOF expression into that of the $\SURE$ \citep[Theorem~1]{stein1981estimation}, we get~\eqref{eq:sure}.

\subsection{Proof of Corollary \ref{cor-particular}} 
\label{sub:proof-cor}
When $\XX=\Id_n$, we have $\XX_I^\ins{T}\XX_I=\Id_I$, which in turn implies that $\Id_I+\lambda \delta_{\xsoly_b(y)} \circ P_{\xsoly_b(y)}$ is block-diagonal. 
Thus, specializing the divergence formula of Theorem~\ref{thm-div} to $\XX=\Id_n$ yields
\begin{align*}
	\widehat{df} & = \tr\bpa{\bpa{ \XXX + \lambda \delta_{\xsoly(y)}  \circ P_{\xsoly(y)}  }^{-1}} \\
	&= \sum_{b \in I} \tr\pa{\pa{\Id_b + \tfrac{\lambda}{\norm{\xsoly_b(y)}}\pa{\Id_b - \tfrac{\xsoly_b(y)\xsoly_b(y)^\ins{T}}{\norm{\xsoly_b(y)}^2}}}^{-1}} \\
&= \sum_{b \in I}\pa{1+\frac{|b|-1}{1+\tfrac{\lambda}{\norm{\xsoly_b(y)}}}} 
\end{align*}
where the last equality follows from the fact that $\mathrm{Proj}_{\xsoly(y)_b^\bot}=\Id_b - \tfrac{\xsoly_b(y)\xsoly_b(y)^\ins{T}}{\norm{\xsoly_b(y)}^2}$ is the orthogonal projector on a subspace of dimension $\abs{b} - 1$.

Furthermore, for $\XX=\Id_n$, $\xsoly_b(y)$ has a closed-form given by block soft thresholding
\begin{equation}\label{eq:block-soft}
  \xsoly_b(y) = 
  \begin{cases}
      0 & \text{if } \norm{y_b} \leq \lambda \\
      (1 - \frac{\lambda}{\norm{y_b}}) y_b & \text{otherwise}
    \end{cases} .
\end{equation}
It then follows that
\begin{align*}
\frac{1}{1+\tfrac{\lambda}{\norm{\xsoly_b(y)}}} = \frac{1}{1+\tfrac{\lambda}{\norm{y_b}-\lambda}} = 1 - \frac{\lambda}{\norm{y_b}} ~.
\end{align*}
Piecing everything together, we obtain
\begin{align*}
\widehat{df} = \sum_{b \in I}\pa{1+(|b|-1)\pa{1 - \frac{\lambda}{\norm{y_b}}}} = \sum_{b \in I} |b| -  \lambda \sum_{b \in I} \frac{|b|-1}{\norm{y_b}} ~.
\end{align*}
As $|I|=\sum_{b \in I} |b|$, we get the desired result. Note that this result can be obtained directly by differentiating~\eqref{eq:block-soft}.

\ifREL
\subsection{Proof of Proposition~\ref{prop:rel}}
Let's introduce the shorhand notation for the reliability 
\eq{
	R = \EE_w\left[\pa{\SURE(\widehat{\mu}(y))-\SE(\widehat{\mu}(y))}^2\right].
} 
Applying \cite[Theorem~4]{vaiter-local-behavior}, we get
\eq{
R = 2n\si^4 - 4\si^4 \tr\pa{{X_IB(y,\lambda)X_I^\ins{T}} \bpa{2\Id_n - X_IB(y,\lambda)X_I^\ins{T}}} + 4\si^2 \EE_w\pa{\norm{\widehat{\mu}(y)-\mu_0}^2}
}
where $B(y,\lambda)=\bpa{ \XXX + \lambda \delta_{\xsoly(y)}  \circ P_{\xsoly(y)}  }^{-1}$ is positive definite by Lemma~\ref{lem-invertible}.

Let's bound the last term. By Jensen's inequality and the fact that $\xsoly(y)$ is a (global) minimizer of \lasso, we have
\begin{align*}
\EE_w\pa{\norm{\widehat{\mu}(y)-\mu_0}^2} 
&\leq 2 (\EE_w\pa{\norm{y - \widehat{\mu}(y)}^2} + \EE_w\pa{\norm{y - \mu_0}^2} \\
&\leq 4 \EE_w\pa{\tfrac{1}{2}\norm{y - \widehat{\mu}(y)}^2+\la \sum_{b \in \Bb} \norm{\xsoly_b(y)}} + 2n\si^2 \\
&\leq 4 \EE_w\pa{\tfrac{1}{2}\norm{y}^2} + 2 n\si^2 =  2 \norm{\mu_0}^2+4n\si^2 ~.
\end{align*}
Let's turn to the second term. We have
\begin{align*}
\tr\pa{X_IB(y,\lambda)X_I^\ins{T}X_IB(y,\lambda)X_I^\ins{T}} &= \tr\pa{X_I^\ins{T}X_IB(y,\lambda)X_I^\ins{T}X_IB(y,\lambda)} \\
&= \norm{X_I^\ins{T}X_IB(y,\lambda)}_F^2 \leq n \norm{X_I^\ins{T}X_IB(y,\lambda)}^2 ~.
\end{align*}
In addition, $X_IB(y,\lambda)X_I^\ins{T}$ is semidefinite positive and therefore
\begin{align*}
R &\leq 2n\si^4 + 4\si^4 \EE_w\pa{\tr\pa{{X_IB(y,\lambda)X_I^\ins{T}} X_IB(y,\lambda)X_I^\ins{T}}} + 4\si^2 \EE_w\pa{\norm{\widehat{\mu}(y)-\mu_0}^2} \\
  &\leq 2n\si^4 + 4n\si^4\EE_w\pa{ \norm{X_I^\ins{T}X_IB(y,\lambda)}^2} + 16n\si^4 + 8\si^2\norm{\mu_0}^2 ~,
\end{align*}
whence we get the desired bound after dividing both sides by $n^2\sigma^4$.

\fi

